\pdfoutput=1
\documentclass[10pt, conference, compsocconf, final]{IEEEtran}

\usepackage{color}                  
\usepackage[english]{babel}         
\usepackage{amsfonts}               
\usepackage{amssymb}                
\usepackage{amsmath}
\usepackage{float}                  
\usepackage{listings}
\usepackage{fancyhdr}
\usepackage{amsthm}
\usepackage{algorithmic}
\usepackage[ruled,vlined]{algorithm2e}
\usepackage{cite}
\usepackage{paralist}

\usepackage{pdftricks}
\begin{psinputs}
  \usepackage{multido}
  \usepackage{pst-all}
  \usepackage{pst-grad} 
  \usepackage{pst-plot} 
  \usepackage{epsfig}
\end{psinputs}
\usepackage[anchorcolor=blue,citebordercolor={1 1 1},citecolor=blue,colorlinks=true]{hyperref}
\theoremstyle{definition}
\newtheorem{property}{Property}

\newtheorem{theorem}{Theorem}
\newtheorem{definition}{Definition}
\newtheorem{lemma}{Lemma}

\newtheorem{corollary}{Corollary}

\usepackage{ifthen} 

\fancypagestyle{temp}{\cfoot{page \thepage\hspace{1pt}}}
\newcommand*{\MyDef}{\mathrm{def}}
\newcommand*{\eqdef}{\ensuremath{\mathop{\overset{\MyDef}{=}}}}
\newcommand*{\nctiw}{{\bar{\texttt{I}}^\texttt{NC}}}
\newcommand*{\citiw}{{\bar{\texttt{I}}^\texttt{CI}}}
\newcommand*{\difftiw}{{\bar{\texttt{I}}^\texttt{DIFF}}}
\newcommand*{\upinter}{\bar{\texttt{I}}}
\newcommand*{\inter}{\texttt{I}}

\newfloat{script}{thp}{lop}

\definecolor{scriptbg}{rgb}{0.9,0.9,0.9}			%
\definecolor{colKeys}{rgb}{0,0,1}			%
\definecolor{colIdentifier}{rgb}{0,0,0}			
\definecolor{colComments}{rgb}{0,0.5,1}			
\definecolor{colString}{rgb}{0.6,0.1,0.1}		%

\lstdefinelanguage{rbacpol}
{
    keywords={subject,role,user},
    sensitive=true,
}

\begin{document}

\title{\LARGE
	Improving Mixed-Criticality System Consistency and Behavior on Multiprocessor Platforms by Means of Multi-Moded Approaches
}

\author{Fran\c cois Santy \hspace{40 mm} Geoffrey Nelissen \hspace{40mm} Jo\"el Goossens \\[0.1in]
	\begin{minipage}{.38\linewidth}
	\begin{center}
	PARTS Research Center\\
	Universit\'e Libre de Bruxelles (ULB)\\
	Brussels, Belgium
	\end{center}
	\end{minipage}
}

\maketitle

\begin{abstract}
Recent research in the domain of real-time scheduling theory has tackled the problem of scheduling mixed-criticality systems upon uniprocessor or multiprocessor platforms, with the main objective being to respect the timeliness of the most critical tasks, at the expense of the requirements of the less critical ones. In particular, the less critical tasks are carelessly discarded when the computation demand of (some of) the high critical tasks increases. This might nevertheless result in system failure, as these less critical tasks could be accessing data, the consistency of which should be preserved. In this paper, we address this problem and propose a method to cautiously handle task suspension. Furthermore, it is usually assumed that the less critical tasks will never be re-enabled once discarded. In this paper, we also address this concern by proposing an approach to re-enable the less critical tasks, without jeopardizing the timeliness of the high critical ones. The suggested approaches apply to systems having two or more criticality levels.
\end{abstract}

\section{Introduction}\label{intro}
The current trend in embedded systems is towards collocating multiple functionalities on shared resources, as illustrated by industrial initiatives in both the aerospace and automotive industry. Nowadays, this trend is even getting stronger with the introduction of multiprocessor and/or multicore architectures. With such collocation however, it is unlikely that all functionalities share the same level of importance (\emph{criticality}), and the timeliness of some functionalities might appear more important than others. In this context, mandatory certification of whole (or subset of) the system by statutory Certification Authorities (CAs) might be required to guarantee the correct behavior of the latter. When certifying a subset of the functionalities, the CAs have to make assumptions about the Worst-Case Execution Time (WCET) of tasks. In practice, it can be observed that the more a task will be assumed critical, the more will it be certified using strongly pessimistic assumptions about its behavior at run-time. In this context, \emph{mixed-criticality} (MC-) systems are an attempt to model systems that need to be certified using \emph{various} assurance levels. The increased pessimism assumed for the more critical tasks during the certification process however introduces a significant over-provisioning of resources at run-time for the latter, having as consequence that the scheduling of mixed-criticality systems must simultaneously deal with two contradictory goals:
\begin{enumerate}
	\item On one hand, and since each task is supposed to carry out some useful computation for the system, it is desirable to respect the timeliness of all tasks, whatever their criticality is, which means that the over-provisioning of resources required during the certification of the more critical tasks could be untighten somehow at run-time;
	\item On the other hand, and since the higher the criticality of a task is, the more dramatic the consequences related to the failure of respecting their timeliness could be, it seems rather cautious to strictly respect the over-provisioning of resources of the more critical tasks, even if this means to interfere with the timeliness of the less critical ones (which is as meaningless as the criticality of these is low).
\end{enumerate}
In practice, these goals are not necessarily exclusive though, and it is possible to pursue them both, at least to some extent. Indeed, beyond vouching for the correctness of the whole (or part of) the system, the certification process also allows for specifying for how long the functionalities can execute concurrently without interfering on each others temporal constraints. Indeed, when a task is certified by some CA, it is assuming a behavior, for the other tasks of the system, the pessimism of which depends on the criticality of the task(s) being certified. Consequently, as long as these assumptions hold at run-time, all the tasks that were certified up to that assurance level are guaranteed to meet their deadline. However, if a more critical task violates these assumptions, by executing longer than was assumed during the certification process, then the timeliness of the whole system can no longer be guaranteed anymore, which results in suspending the tasks that are less critical to ensure the timeliness of the tasks being more critical. In the state-of-the-art literature, the period during which these tasks are prevented from executing, henceforth called the \emph{suspension delay}, is considered to be unlimited. \\
With little thought, one can notice that mixed-criticality systems can be seen as systems that exhibit \emph{multiple behaviors}, issued from several operating modes, where each of these operating modes is characterized by a given set of functionalities. Such systems are commonly referred to, in the real-time literature, as \emph{multi-moded} systems~\cite{Real:2004:MCP:969960.969963}. When viewing mixed-criticality systems as multi-moded systems, the transition from one operating mode to the other reduces the set of tasks performed by the system. Indeed, less critical tasks are suspended, but new tasks are never introduced (at least in the traditional way of handling mixed-criticality systems, as we will present an approach which allows to overstep this restriction). The transition from one mode to the other is triggered when the system detects that the certifications assumptions does not hold anymore.\\
When considering the temporal robustness of the system, mixed-criticality is a natural approach, since it aims at favoring the timeliness of the functionalities that are crucial for the system. Nevertheless, discarding carelessly tasks that are less critical may have severe consequences on the system's consistency. Indeed, a task that is executing could be accessing data, the integrity and consistency of which has to be preserved. Under these assumptions, it is strongly undesirable to kill a job incautiously. It would therefore be wiser to allow the jobs released by less critical tasks to complete their execution when the transition from one operating mode to the other is triggered, but this transition stage might introduce a transient overload which should have no impact on the timeliness of high criticality tasks. Furthermore, while task suspension is carried out with concern for the respect of the timeliness of the high critical tasks, a suspended functionality can no longer handle the task it is in charge of. Consequently, an everlasting suspension delay might appear as strongly undesirable, and we might wish to be able to re-enable a suspended task at some point during the execution of the system.

\subsection{Related Work}\label{related}
Mixed-Criticality scheduling is a research domain initially introduced by Vestal~\cite{Vestal:2007:PSM:1338441.1338659}. Nowadays, the Mixed-Criticality (MC)-Schedulability problem is known to arise in the context of applications being subject to multiple certification requirements. Many work has addressed the problem for systems implemented upon \emph{uniprocessor} platforms \cite{Li:2010:LSA:1879021.1879035,Baruah:2011:MSS:2040572.2040633,R:Baruah:2011a,g:p:m:w,DBLP:conf/ecrts/BaruahBDLMSS12,Santy:2012:RMS:2354411.2355225}. More recently, research has been oriented towards the study of mixed-criticality scheduling upon multiprocessor platforms~\cite{10.1109/ECRTS.2012.29}\cite{DBLP:conf/ecrts/LiB12}. Let us highlight that many of these work focused on a restricted and easier case of mixed-criticality systems, called dual-criticality systems, presenting only two distinct criticality levels. Nevertheless, to our knowledge, no work has tackled the problem that arises when mixed-criticality real-time tasks are incautiously discarded. Though in a previous work~\cite{Santy:2012:RMS:2354411.2355225}, we highlighted the fact that task suspension was often carried out too early, and that the system's computational resources could be exploited more cleverly. We also addressed the problem of managing task re-enablement. However, our work focused only uniprocessor platforms, while this work will be applied to the more general case of identical multiprocessor platforms.

\subsection{This Research}\label{ourWork}
We believe that a multi-mode approach can bring some insight on how the problems mentioned in Section~\ref{related} could be formalized. To our knowledge, mixed-criticality systems have never formally been defined as multi-moded systems. In this research, we thus seek to establish a formal link between mixed-criticality and multi-moded hard real-time systems upon identical multiprocessor platforms. We will then show how to safely complete the execution of jobs released by the less critical tasks, potentially after their deadline, when the computational demand of the more critical tasks in the systems increases. In our opinion, it is safer in terms of system integrity to complete a job beyond its deadline, instead of dropping it carelessly once the latter is reached. We will also show how to safely reduce the suspension delay suffered by the less critical tasks as the system load decreases, by allowing the re-enablement of the latter. This problem was almost never tackled in the current literature, but still presents a significant interest, since the less critical tasks can improve the overall behavior of the system. Finally, since the notion of mixed-criticality systems initially did not put any restriction on the number of criticality levels, and has been implemented using various degrees of criticality by industrial standards\footnote{The RTCA-DO178B standard, for aerospace, defines 5 criticality levels, while the ISO26262, in the automotive industry, defines 4 criticality levels.}, our approach can be applied to task sets having any number of criticality level. This is an attractive feature, as it is much more compliant with current industrial standards.

\section{Model and Definitions}\label{modelDef}
\subsection{Mixed-criticality Specification}\label{subsec:model}
We consider multiprocessor platforms composed of a fixed number $m$ of identical processors, denoted by $\pi = \{\pi_1, \pi_2, ..., \pi_m\}$. Furthermore, we consider mixed-criticality sporadic task sets $\tau = \{\tau_1, \tau_2, ..., \tau_n\}$, where the maximum criticality of a task is bounded by a natural value which we denote by $\mathnormal{L}$. A task $\tau_i$ in such a system is characterized by a 4-tuple of parameters $\{T_i, D_i, L_i, C_i\}$ where:
\begin{compactitem}
	\item[$\bullet$] $T_i \in \mathbb{N}_0$ is the minimum inter-arrival time separating two consecutive activations of task $\tau_i$;
	\item[$\bullet$] $D_i \in \mathbb{N}_0$ is the deadline of task $\tau_i$, with $D_i \leq T_i$;
	\item[$\bullet$] $L_i \in \mathbb{N}_0$ is the criticality of the task $\tau_i$, with $L_i \leq \mathnormal{L}$;
	\item[$\bullet$] $C_i \in \mathbb{N}_0^{\mathnormal{L}}$ is a size $L$ vector of WCET, where $C_i(\ell)$ is an estimation of the WCET of task $\tau_i$ at criticality level $\ell \in [1,\mathnormal{L}]$.
\end{compactitem}
Given these parameters, each task $\tau_i$ will generate a potentially infinite sequence of jobs, each release being separated by at least $T_i$ time units, and each job having a hard deadline $D_i$ time units after its release. We assume $C_i(\ell)$ is monotonically increasing for increasing values of $\ell$. More precisely, for task $\tau_i$ the following two conditions hold:
\begin{compactitem}
	\item[$\bullet$] $\forall \ell \in [1, L_i)$, $C_i(\ell) \leq C_i(\ell+1)$;
	\item[$\bullet$] $\forall \ell \in [L_i, \mathnormal{L}]$, $C_i(\ell) = C_i(L_i)$.
\end{compactitem}
It follows that no task is supposed to execute longer than its WCET at its own criticality level. The $k^\text{th}$ job $J_{i,k}$ released by a mixed-criticality task $\tau_i$ is characterized by a 3-tuple $\{r_{i,k}, d_{i,k}, c_{i,k}\}$ where:
\begin{compactitem}
	\item[$\bullet$] $r_{i,k} \in \mathbb{N}$ is the time instant at which $J_{i,k}$ was released. Since we consider sporadic task systems, we have $r_{i,k} - r_{i,k-1} \geq T_i$, and $r_{i,1} \geq 0$;
	\item[$\bullet$] $d_{i,k} \in \mathbb{N}_0$ is the absolute deadline of $J_{i,k}$. More precisely, $d_{i,k} \eqdef r_{i,k} + D_i$;
	\item[$\bullet$] $c_{i,k} \in \mathbb{N}_0$ is the exact execution time of $J_{i,k}$. From the specifications of $\tau_i$, we can say that $c_{i,k} \leq C_i(L_i)$, but the exact value of $c_{i,k}$ will not be known until $J_{i,k}$ completes its execution;
	\item[$\bullet$] $f_{i,k} \in \mathbb{N}_0$ is the absolute time at which $J_{i,k}$ finishes its execution. Again, this value will only be known once $J_{i,k}$ actually completes its execution.
\end{compactitem}

\begin{definition}[Available job]
At any time $t$, we call the $k^\text{th}$ job $J_{i,k}$ released by task $\tau_i$ \emph{available} if $t \geq r_{i,k}$ and $J_{i,k}$ has not yet completed its execution.
\end{definition}

The actual execution time of job $J_{i,k}$ is not known from the specification of $\tau_i$, but will only be discovered when $J_{i,k}$ completes its execution. Besides, the behavior of task $\tau_i$ might change from one execution to the other, so the actual executions times of the sequence of jobs released by $\tau_i$ may vary, leading to introduces the notion of \emph{task scenario}.

\begin{definition}[Task scenario]
We define the scenario $s_i^t$ of task $\tau_i$ at time $t$ as the set of exact execution times $\{c_{i,1}, ..., c_{i,k}\}$ for each of the $k$ ($k$ being dependent of $t$) jobs released by $\tau_i$ that already completed their execution at time $t$.
\end{definition}

\begin{definition}[Task set scenario]
We define the scenario of the mixed-criticality task set $\tau$ at time $t$ as $s^t = \{s_1^t, ..., s_n^t\}$.
\end{definition}

\begin{definition}[$\ell$-interval]
Given a task set scenario, an $\ell$-interval is an interval $[t_a,t_b)$ such that at time $t_a$, the system switched at criticality level $\ell$, and the criticality of the system remained at level $\ell$ until time $t_b$.
\end{definition}

\begin{definition}[Worst-case response time]
The worst-case response time (WCRT) $R_i(\ell)$ of task $\tau_i$ at criticality level $\ell$ is the maximum duration to execute any job of $\tau_i$, in any scenario of criticality $\ell$.
\end{definition}

\subsection{Modeling Mixed-Criticality in Terms of Multi-Mode}\label{MCtoMM}
In this section, we propose a formalization of sporadic real-time mixed-criticality tasks in terms of multi-mode tasks. Each such task can be seen as being composed of different execution modes, and can be represented by means of a tuple, as described below:
\begin{align*}
	\tau_i = \{	&\tau_i^{M_1} = \{T_i^{M_1}, D_i^{M_1}, C_i^{M_1}\},... \\
			&\tau_i^{M_\mathnormal{L}} = \{T_i^{M_\mathnormal{L}}, D_i^{M_\mathnormal{L}}, C_i^{M_\mathnormal{L}}\}\}
\end{align*}
where $M_1, ..., M_\mathnormal{L}$ represent the different operating modes to which the task belongs, and $X_i^{M_y}$ represents the value of parameter $X$ of task $\tau_i$ when $\tau_i$ is executed in the operating mode $M_y$. Consequently, a natural representation for the mixed-criticality task $\tau_i = \{T_i, D_i, L_i, \{C_i(1), ..., C_i(L_i)\}\}$ is by means of the following tuple:
\begin{align*}
	\tau_i = \{	\tau_i^{1} &= \{T_i, D_i, C_i(1)\},... \\
			\tau_i^{L_i} &= \{T_i, D_i, C_i(L_i)\}\}	
\end{align*}
The interpretation of the tuple is the following: when the system is executing in operating mode $M_{\ell_1}$, task $\tau_i$ is executed according to the parameters given by $\tau_i^{\ell_1}$. The scheduler thus makes the assumption that $\tau_i$ will not generate a job the execution time of which will exceed $C_i(\ell_1)$. If $\tau_i$ does release a job whose execution exceeds $C_i(\ell_1)$, then a transition from one operating mode to another occurs. If the system switches from mode $M_{\ell_1}$ to mode $M_{\ell_2}$, then $\tau_i$ is executed according to the parameters given by $\tau_i^{\ell_2}$. From that specific time onward, the scheduler consequently assumes that $\tau_i$ will not generate jobs whose executions will exceed $C_i(\ell_2)$. A task $\tau_i$ belongs to every operating mode up to its own criticality level. More precisely, $\tau_i \in M_\ell \Leftrightarrow \ell\leq L_i$. A task $\tau_i$ will be enabled in every operating mode to which it belongs, and will be suspended in every other operating mode.
\begin{definition}[Enabled/suspended task]
At run-time, a task $\tau_i$ is said to be enabled if $\tau_i$ can generate new jobs that are dispatched by the scheduler. Otherwise, $\tau_i$ is said to be suspended. 
\end{definition}
In the following, we will consider that the system is running is mode $M_\ell$ at time $t$ if and only if all tasks of criticality less than $\ell$ are disabled, no job of criticality less than $\ell$ is still available at time $t$, and every task of criticality greater than or equal to $\ell$ is enabled. Notice that at system start-up, the system is running in operating mode $M_1$, which implies that every task is enabled. In the remainder of the paper, we will denote by $\tau^\ell=\{\tau_{\ell_1},...,\tau_{n_\ell}\} \subseteq \tau$ the set of tasks belonging to operating mode $M_\ell$, i.e.\@ the set of tasks $\tau_j$ with $L_j \geq \ell$.

\subsection{Presentation of the $\mathsf{MSM}$ Scheduler}\label{subsec:scheduler}
In Section~\ref{subsec:model}, we introduced the concept of criticality level of the scenario. This concept allows to define what is considered as being a feasible schedule when considering mixed-criticality systems.
\begin{definition}[Feasible schedule]
A schedule for a scenario $s^t$ is feasible if, during every $\ell$-interval, $1\leq\ell\leq L$, every job $J_{i,k} \, | L_i \geq \ell$ ($1 \leq i \leq n$) that completed its execution by time $t$, received execution time $c_{i,k}$ between $r_{i,k}$ and $d_{i,k}$.
\end{definition}
This definition implies that mixed-criticality scheduling is only concerned with respecting temporal constraints of tasks the criticality of which is higher than or equal to the criticality level of the scenario.

\begin{definition}[$\mathcal{S}$-Schedulable]
Let $\mathcal{S}$ be a scheduling policy, and $\tau$ a mixed-criticality task set. We say $\tau$ is $\mathcal{S}$-Schedulable if, for \emph{any} scenario of $\tau$, $\mathcal{S}$ generates a feasible schedule.
\end{definition}
Since an on-line scheduling policy discovers the exact execution time of the jobs when they complete their execution, the criticality level of the scenario is not known beforehand. In the state-of-the-art way of handling mixed-criticality systems, as soon as a job exceeds its WCET of level $\ell$, the criticality of the scenario is raised to level $\ell+1$, all available jobs of criticality lower than $\ell+1$ are \emph{dropped}, and future releases from tasks of which the criticality is lower than $\ell+1$ are no longer taken into consideration. At this point, recall that the goal of this research is to overstep these restrictions. Finally, the following definition specifies what is considered as being an MC-Schedulable system.

\begin{definition}[MC-Schedulable]
A mixed-criticality task set $\tau$ is MC-schedulable if there exists a scheduling policy $\mathcal{S}$ such that $\tau$ is $\mathcal{S}$-schedulable.
\end{definition}

In the remainder of this paper, we consider the multiprocessor \emph{global} \emph{preemptive} \emph{work-conservative} \emph{static-priority} scheduler $\mathsf{MSM}$ proposed by Pathan~\cite{10.1109/ECRTS.2012.29}, with the following interpretation:
\begin{compactitem}
	\item[$\bullet$] at each time instant, a global scheduler dispatches the $m$ highest priority jobs (if any) on the $m$ processors of the platform;
	\item[$\bullet$] a preemptive scheduler reserves the right to interrupt a job $J_i$ belonging to a task $\tau_i$, that is executing on a given processor, to assign this processor to another available job $J_j$ belonging to a task $\tau_j$, with $i\neq j$;
	\item[$\bullet$] a work-conservative scheduling strategy never keeps a processor idle when there are available jobs;
	\item[$\bullet$] a static-priority scheduler assigns priorities to the tasks of the system, each job being scheduled inheriting from the priority of the task that released it.
\end{compactitem} 
The algorithm $\mathsf{MSM}$ thus dispatches available jobs according to traditional global static-priority scheduling, but has two additional implementation features, namely \emph{job execution monitoring}, in order to detect a transition to a higher criticality level, and \emph{task suspension}, to drop the less critical tasks when such a transition is detected. To our knowledge, $\mathsf{MSM}$ is the only multiprocessor scheduler that can be applied to mixed-criticality task sets having any number of criticality levels. Furthermore, and as already motivated by Pathan~\cite{10.1109/ECRTS.2012.29}, static-priority schedulers are generally preferred by industries, as their decisions are predictable, thus enforcing the reliability property of the real-time systems they schedule. \\
The priority ordering of a task set is constructed according to Audsley's Optimal Priority Assignment (OPA) approach~\cite{Audsley_1991}, according to the response time analysis of the system. These notions are further discussed in Sections~\ref{sec:wcrt} and~\ref{sec:audsley}. In the remainder of the paper, we will denote by $\texttt{hp}(\tau_i)$ the set of tasks $\tau_j$ having a priority higher than $\tau_i$, and by $\texttt{lp}(\tau_i)$ the set of tasks $\tau_j$ having a priority lower than $\tau_i$.

\subsubsection{The Worst-Case Response Time Computation}\label{sec:wcrt}~\\
This section briefly introduces the work proposed by Pathan~\cite{10.1109/ECRTS.2012.29}. An upper-bound on the WCRT $R_i(\ell)$ of task $\tau_i$ at criticality level $\ell \in [1, L_i]$ is computed by considering the execution of any job $J_{i,k}$ released by $\tau_i$ in a window starting at $J_{i,k}$'s release time $r_{i,k}$, and finishing $\Delta$ time units later, i.e.\@ at time instant $r_{i,k}+\Delta$. The procedure aims at defining an upper-bound on the \emph{workload}, the \emph{interfering workload}, the \emph{total interfering workload}, and the \emph{interference} (as explained in the following definitions) suffered by task $\tau_i$ in any scenario of criticality level $\ell$ during that interval.

\begin{definition}[Workload]\label{def:workload}
The \emph{workload} of a higher priority task $\tau_j$ within the window of size $\Delta$ is the cumulative length of time interval during which the jobs released by $\tau_j$ execute within the window.
\end{definition}

A task $\tau_j\in\texttt{hp}(\tau_i)$ is considered as a \emph{carry-in} task if $\tau_j$ released a job before the start of the window, and that job executes (partially or fully) within the window. Otherwise, $\tau_j$ is considered as a \emph{non carry-in} task. Pathan~\cite{10.1109/ECRTS.2012.29} highlighted the fact that if a higher-priority task is a carry-in task, then its worst-case interference on the lower-priority task is higher than if it was non carry-in.

\begin{definition}[Carry-in/Non Carry-In Interfering Workload]\label{def:IntWorkload}
The \emph{carry-in interfering workload} $\citiw_{j,i}(\Delta,\ell)$ (resp. \emph{non carry-in interfering workload} $\nctiw_{j,i}(\Delta,\ell)$) of $\tau_j$ on task $\tau_i$ is the cumulative length of time interval during which a job released by a carry-in task $\tau_j$ (resp. a non carry-in task $\tau_j$) executes, and $J_{i,k}$ not dispatched on any processor.
\end{definition}

The difference between the carry-in and non carry-in interfering workload of a task $\tau_j\in\texttt{hp}(\tau_i)$ will be denoted by $\difftiw_{j,i}(\Delta,\ell) \eqdef \citiw_{j,i}(\Delta,\ell)-\nctiw_{j,i}(\Delta,\ell)$

\begin{definition}[Total Interfering Workload]\label{def:totIntWorkload}
The \emph{total interfering workload} $\upinter_i(\Delta, \ell)$ is the sum of interfering workload of all the higher priority tasks within the window.
\end{definition}
The $\mathsf{MSM}$ algorithm computes the total interfering workload of task $\tau_i$ as follows:
	\begin{equation}\label{eq:upperBoundTotIntWorkload}
		\upinter_i(\Delta, \ell) = \sum_{\tau_j\in\texttt{hp}(\tau_i)} \nctiw_{j,i}(\Delta,\ell) + \sum_{\tau_j\in\texttt{hp}_{m-1}(\tau_{i})} \difftiw_{j,i}(\Delta,\ell)
	\end{equation}
where $\texttt{hp}_{m-1}$ is the set of at most $m-1$ carry-in tasks belonging to $\texttt{hp}(\tau_i)$ that have the largest value of $\difftiw_{j,i}(\Delta,\ell)$.

The reason to consider at most $m-1$ carry-in tasks comes from a discussion from Guan \emph{et al.}\@~\cite{Guan:2009:NRT:1683310.1684939}, formalized in the following property.

\begin{property}\label{property:guan}
The total interfering workload is upper-bounded by considering at most $m-1$ carry-in tasks within the window of any lower priority task, when considering global static priority scheduling of constrained-deadline sporadic task sets.
\end{property}

\begin{definition}[Interference]\label{def:interference}
The interference suffered by a task $\tau_i$ in any scenario of criticality level $\ell$, and during a time interval of length $\Delta$, is the cumulative length of time interval during which the $m$ processors are busy executing tasks belonging to $\texttt{hp}(\tau_i)$. An upper-bound on the interference suffered by $\tau_i$ over an interval of length is given by $\left\lfloor \frac{\upinter_{i}(\Delta, \ell)}{m} \right\rfloor$.
\end{definition}

Finally, and from the above definitions, since in any scenario of criticality level $\ell$, $J_{i,k}$ is allowed to execute for at most $C_i(\ell)$ time units, the WCRT $R_i(\ell)$ of task $\tau_i$ at criticality level $\ell$ is obtained by determining the least fixed point of the following function:
	\begin{equation}
		R_i(\ell) = C_i(\ell) + \left\lfloor \dfrac{\upinter_{i}(R_i(\ell), \ell)}{m} \right\rfloor
	\end{equation}
	
Since $\upinter_i(\Delta,\ell)$ is an upper-bound on the total interfering workload suffered by task $\tau_i$, the \emph{actual} total interfering workload suffered by $\tau_i$ over an interval of length $\Delta$, in any scenario of criticality level $\ell$, will be denoted by $\inter_{i}^*(\Delta,\ell)$.

\subsubsection{Finding Priorities Using Audsley's Approach}\label{sec:audsley}~\\
The response time analysis described in Section~\ref{sec:wcrt} is used to find a static-priority ordering of the mixed-criticality task set $\tau$, as depicted by Algorithm~\ref{alg:priorityAssignment}. Indeed, at each step, the method tries to identify a task $\tau_i$ satisfying $R_i(\ell) \leq D_i$, $1 \leq \ell \leq L_i$\footnote{This is due to the fact that to our knowledge, there is no proof that $R_i(\ell) \leq R_i(\ell+1)$, $1\leq\ell<L$ holds.}. If such a task is found, then the same reasoning is iteratively applied to task set $\tau\setminus\{\tau_i\}$. Otherwise, if for every task $\tau_i$, $\exists \ell \, | \, R_i(\ell) > D_i$, then the method fails.
\setlength{\textfloatsep}{0pt}
	\begin{algorithm}
	$\mathsf{pr} \leftarrow |\tau|$ \\
	\While{$\tau \neq \varnothing$}{
		Let $\tau_i$ be a task from $\tau$\;
		\If{$R_i(\ell) \leq D_i$, $1\leq\ell\leq L_i$}{
			Assign $\tau_i$ the priority $\mathsf{pr}$\;
			$\tau \leftarrow \tau\setminus\{\tau_i\}$\;
			$\mathsf{pr} \leftarrow \mathsf{pr}-1$\;
		}
		\Else{
			\textbf{return} error\;
		}
	}
	\caption{OPA Algorithm}\label{alg:priorityAssignment}
	\end{algorithm}

\subsection{Mode Transition Specifications}\label{subsec:mts}
A change in the operating mode of the system is instantiated whenever the system detects a change in its internal state. More precisely, switching from criticality level $\ell_1$ to criticality level $\ell_2$ can be seen as switching from one operating mode, referred to as the \emph{old}-mode $M_{\ell_1}$, to another operating mode, referred to as the \emph{new}-mode $M_{\ell_2}$. A \emph{Mode Change Request} (\emph{MCR}) is defined as being the event that triggers a mode change transition. 
In the current state-of-the-art literature, the only MCR that is considered is triggered whenever the system detects that a task $\tau_i$ exceeds its WCET at criticality level $\ell$, $1 \leq \ell < L_i$, resulting in the system to switch from operating mode $M_\ell$ to operating mode $M_{\ell+1}$. In the remainder of this paper, this MCR will be referred to as an \emph{increasing mode change request} denoted by $\texttt{IMCR}_{\ell+1}$. 
\begin{definition}[Increasing Mode Change Request]
Whenever the system is running in operating mode $M_\ell$, an increasing mode change request from criticality level $\ell$ to criticality level $\ell+1$, $\forall \ell \in [1, L)$, denoted by $\texttt{IMCR}_{\ell+1}$, is the event that will result in switching from operating mode $M_{\ell}$ to operating mode $M_{\ell+1}$. The time at which $\texttt{IMCR}_{\ell+1}$ is triggered will be denoted by $t_{\texttt{IMCR}_{\ell+1}}$.
\end{definition}

To handle this transition, the traditional approach that is currently implemented in the state-of-the-art literature consists in suspending \emph{instantaneously} and \emph{forever} tasks that do not belong to the new mode, their potential available jobs at time $t_{\texttt{IMCR}_{\ell+1}}$, henceforth called the rem-jobs, not even being dispatched anymore. The transition stage is thus trivial, and the system instantaneously switches from operating mode $M_\ell$ to operating mode $M_{\ell+1}$. By suspending the less critical tasks upon an $\texttt{IMCR}_{\ell+1}$, the goal is to respect both the ($\ell$+1)-periodicity and ($\ell$+1)-feasibility, defined as follows. 

\begin{definition}[$\ell$-periodicity, adapted from~\cite{Real:2004:MCP:969960.969963}]
The $\ell$-periodicity is the property that requires the activation pattern of each task $\tau_i$ such that $L_i \geq \ell$ to be respected. In other words, the activation pattern of a task $\tau_i$ should not be altered by increasing mode change requests up to $\tau_i$'s own criticality.
\end{definition}

\begin{definition}[$\ell$-feasibility, adapted from~\cite{Real:2004:MCP:969960.969963}]
The $\ell$-feasibility is the property that requires each task $\tau_i$ such that $L_i \geq \ell$ to meet its deadline.
\end{definition} 

Consequently, an increasing mode change request $\texttt{IMCR}_{\ell+1}$ $\forall \ell < L_i$ should have no impact on the deadline of a task $\tau_i$. With concern for the respect of these properties, the traditional approach mentioned above however introduces two major drawbacks. On one hand, and since the tasks belonging to the old-mode might be accessing data, it might appear more cautious to allow them to complete their execution before discarding them. This nevertheless could result in a temporary overload which could jeopardize the timeliness of the new-mode tasks. The transition should therefore be managed in such a way that the ($\ell$+1)-periodicity and ($\ell$+1)-feasibility properties are preserved, while the rem-jobs could be allowed to complete their execution. Furthermore, as an everlasting suspension might not be necessary to keep on respecting the ($\ell$+1)-periodicity and ($\ell$+1)-feasibility, we would also like to re-enable the less critical tasks as soon as possible. Achieving the second goal nevertheless calls for the use of a new mode change request, namely a \emph{decreasing mode change request}. 

\begin{definition}[Decreasing Mode Change Request]
Whenever the system is running in operating mode $M_{h}$, a decreasing mode change request from criticality level $h$ to $\ell$, $\forall \ell \in [1, h)$, denoted by $\texttt{DMCR}_{\ell}$, is the event that will result in switching from operating mode $M_{h}$ to operating mode $M_\ell$. The time at which $\texttt{DMCR}_{\ell}$ is triggered will be denoted by $t_{\texttt{DMCR}_{\ell}}$.
\end{definition}

We again insist on the fact that current research in the field of mixed-criticality never considers the re-enablement of previously suspended tasks, at the exception of~\cite{Santy:2012:RMS:2354411.2355225}. Based on the two mode change requests that are considered, we present in Section~\ref{sec:protocol} mode change protocols whose goals are:
\begin{compactitem}
	\item[$\bullet$] Upon an $\texttt{IMCR}_{\ell+1}$, complete the execution of the rem-jobs as soon as possible, as this will improve system consistency;
	\item[$\bullet$] Upon a $\texttt{DMCR}_{\ell}$, re-enable the previously suspended tasks belonging to mode $M_\ell$ as soon as possible, as this will improve the system's overall behavior.
\end{compactitem} 
\begin{definition}[Valid increasing transition protocol]
Assuming the system is running in operating mode $M_\ell$, an increasing transition protocol $\mathcal{P}$ is said to be \emph{valid} if and only if, upon an $\texttt{IMCR}_{\ell+1}$, $\mathcal{P}$ ensures both the ($\ell+1$)-periodicity and ($\ell+1$)-feasibility.
\end{definition}
\begin{definition}[Valid decreasing transition protocol]
Assuming the system is running in operating mode $M_h$, a decreasing mode transition protocol $\mathcal{P}$ is said to be \emph{valid} if and only if, upon an $\texttt{DMCR}_{\ell}$, $\mathcal{P}$ ensures both the $h$-periodicity and $h$-feasibility.
\end{definition}
It follows that the traditional approach is a valid increasing transition protocol, since upon a $\texttt{IMCR}_{\ell+1}$, and from time $t_{\texttt{IMCR}_{\ell+1}}$ onward, it aims only at respecting the timeliness of tasks the criticality of which is at least equal to $\ell+1$. Furthermore, the former approach does not consider decreasing mode change requests, so the timeliness of the high critical tasks can not be jeopardized by the re-enablement of the less critical tasks. By extending this approach to allow the rem-jobs to complete their execution, and allow the re-enablement of the less critical tasks, we thus only seek to achieve a cleverer usage of the platform on which the mixed-criticality system is executed, without compromising the reliability of the considered mixed-criticality system.
In the literature related to multi-moded systems, it is sometimes assumed that a mode change request $\texttt{IMCR}_{\ell+1}$ may only be triggered in the steady state of the system, i.e.\@ whenever the system is in mode $M_\ell$. In the case of mixed-criticality systems however, a mode change request $\texttt{IMCR}_{\ell+2}$ could happen during the handling of the mode change request $\texttt{IMCR}_{\ell+1}$. Indeed, in mode $M_\ell$, it is assumed no task $\tau_i$ will generate a job the execution time of which will exceed $C_i(\ell)$. If this nevertheless happens, then $\texttt{IMCR}_{\ell+1}$ is triggered, but this does not prevent the job of $\tau_i$ to keep on executing, potentially leading to exceed $C_i(\ell\text{+1})$ and trigger $\texttt{IMCR}_{\ell+2}$ while still handling $\texttt{IMCR}_{\ell+1}$.\\

\section{The Protocol}\label{sec:protocol}

\subsection{Introductory Definitions}
The OPA procedure described in Section~\ref{subsec:scheduler} considers specific scenarios for each tasks of the system, where the execution time of each job $J_{i,k}$ is exactly equal to the WCET of task $\tau_i$ at a given criticality level $\ell$ such that $1 \leq \ell \leq L_i$. These scenarios are called \emph{basic}, as formalized by the following definitions.

\begin{definition}[Basic task scenario]\label{def:tasksetscenario}
At any time $t$, we call the scenario $s_i^t$ of task $\tau_i$ basic if for all $c_{i,k} \in s_i^t$, there exists a criticality level $\ell_k \in [1, L_i]$ such that $c_{i,k} = C_i(\ell_k)$.
\end{definition}

In a basic task scenario, each job released by a task $\tau_i$ will thus complete its execution after having consumed exactly its WCET at a given criticality level (which is not necessarily equal to the task's own criticality level $L_i$).

\begin{definition}[Basic task set scenario]
At any time $t$, we call the scenario of the mixed-criticality task set $\tau$ basic if every $s_i^t \in s^t$ is a basic task scenario.
\end{definition}

\begin{definition}[Basically $\mathcal{S}$-Schedulable]
Let $\mathcal{S}$ be a scheduling policy, and $\tau$ a mixed-criticality task set. We say $\tau$ is basically $\mathcal{S}$-schedulable if, for \emph{any} basic scenario of $\tau$, $\mathcal{S}$ generates a feasible schedule.
\end{definition}

\begin{theorem}[Baruah \emph{et al.}~\cite{DBLP:journals/tc/BaruahBDLMMS12}]\label{theo:basicallySchedulable}
A mixed-criticality task set $\tau$ is MC-schedulable on a given platform $\pi$ if $\tau$ is basically $\mathcal{S}$-schedulable by a given scheduling policy $\mathcal{S}$.
\end{theorem}

The OPA procedure thus aims at determining whether a given mixed-criticality task set $\tau$ is basically $\mathsf{MSM}$-schedulable. Indeed, during the WCRT computation at criticality level $\ell$, $\forall \ell \in [1, L]$, the procedure considers each task $\tau_i$ will release jobs whose execution times will be equal to $\tau_i$'s WCET a some criticality level. The scenarios that are considered are thus basic. Consequently, if the OPA procedure succeeds, $\tau$ is deemed basically $\mathsf{MSM}$-schedulable, and according to Theorem~\ref{theo:basicallySchedulable}, $\tau$ is also MC-schedulable.\\
Throughout this section, we will consider that the priority of a task is given by its index, such that task $\tau_i$ has a higher priority than task $\tau_{i+1}$.

\subsection{Handling Increasing Mode Change Requests}\label{subsec:increasing}
The offline analysis stage, whose goal is to vouch for the feasibility of the considered mixed-criticality system, introduces two sources of pessimism:
\begin{compactitem}
	\item[$\bullet$] The WCET estimation introduces a level of pessimism which is as significant as the criticality of the considered task is high;
	\item[$\bullet$] The WCRT computation, as it is carried out in Section~\ref{sec:wcrt}, leads to defining an \emph{upper}-bound on the WCRT of each task, rather than the actual \emph{exact} WCRT.
\end{compactitem}
The introduction of these sources of pessimism leads to overestimate the actual requirements of the system. In particular, this means that during the execution of the system, a lot of computation resources could be wasted. In this section, we will show how to take advantage of the pessimism introduced during both the WCET and WCRT analysis stages by reclaiming these wasted resources, to allow the less critical tasks to complete their execution during the online scheduling phase.\\
Before diving in the details of our contributions, recall that the WCRT computation assumes that whenever the system switches from operating mode $M_\ell$ to operating mode $M_{\ell+1}$ at time $t_{\texttt{IMCR}_{\ell+1}}$, each task $\tau_i \in \tau$ such that $L_i \leq \ell$, do not interfere anymore on tasks $\tau_j \in \tau$, such that $L_j > \ell$, from time $t_{\texttt{IMCR}_{\ell+1}}$ onward. To prevent this interference when switching from operating mode $M_\ell$ to operating mode $M_{\ell+1}$, the mixed-criticality system model prescribes the immediate suspension of each task satisfying $L_i \leq \ell$, their potential rem-jobs being instantaneously dropped as soon as the mode change request $\texttt{IMCR}_{\ell+1}$ is triggered. As a result, if the main goal is to avoid jeopardizing the timeliness of the tasks the criticality of which is at least equal to $\ell+1$, the completion of the rem-jobs must be handled in such a way that no interference can occur on a task with a criticality at least equal to $\ell+1$. We now present our contributions, by suggesting three alternative approaches to avoid this interference.

\subsubsection{A Naive Approach}~\\
Upon an $\texttt{IMCR}_{\ell+1}$, an intuitive approach consists in assigning the rem-jobs the criticality of which is less than $\ell+1$ a priority that is \emph{lower} than the jobs the criticality of which is at least equal to $\ell+1$. The execution of the rem-jobs thus takes place whenever there are less than $m$ available jobs released by tasks $\tau_j\in\tau^{\ell+1}$. The following theorem proves this approach to be correct, i.e.\@ ($\ell+1$)-periodicity and ($\ell+1$)-feasibility is preserved among an $\texttt{IMCR}_{\ell+1}$.
\begin{theorem}
Upon an $\texttt{IMCR}_{\ell+1}$, assigning the rem-jobs of criticality level $\ell$ a priority that is lower than every task $\tau_i$ such that $L_i \geq \ell+1$ is a valid increasing transition protocol.
\end{theorem}
\begin{IEEEproof}
By assigning the rem-jobs a priority that is lower than every task $\tau_i$ such that $L_i \geq \ell+1$, they cannot interfere on the tasks $\tau_j\in\tau^{\ell+1}$. It follows that the $(\ell+1)$-periodicity and $(\ell+1)$-feasibility are preserved, and that this increasing transition protocol is valid.
\end{IEEEproof}

\subsubsection{A WCET-Based Approach}\label{sec:improvedApproach}~\\
The drawback of the naive approach lies in the fact that upon an $\texttt{IMCR}_{\ell+1}$, the completion of the rem-jobs having a criticality equal to $\ell$ can only occur whenever there are strictly less than $m$ available jobs of criticality at least equal to $\ell+1$. As a consequence, the rem-jobs could complete their execution well beyond their deadline. We can however take advantage of the pessimism introduced during the WCET analysis to allow the rem-jobs to complete their execution earlier. Indeed, recall that the offline analysis vouched for the feasibility of any basic scenario of the system. More precisely, and considering that the system is currently executing in operating mode $M_\ell$, this means that the scheduler will make the assumption that every job $J_{i,k}$ will complete its execution after exactly $c_{i,k} = C_i(\ell)$ time units. Nevertheless, the fact is that $J_{i,k}$ could complete its execution after an amount of time $c_{i,k}$ such that $c_{i,k} \leq C_i(\ell)$, in any scenario of criticality $\ell$. Let us therefore denote by $\texttt{ur}_{i,k}(\ell) \eqdef C_i(\ell) - c_{i,k}$ the unused resources of $J_{i,k}$ in any scenario of criticality level $\ell$. The value $\texttt{ur}_{i,k}(\ell)$ thus represents a fraction of the resources that were reserved by the scheduler for job $J_{i,k}$, under the assumption of a basic scenario of criticality $\ell$, and that $J_{i,k}$ did not consume. \\
Following from the above discussion, we propose an alternative increasing transition protocol, that slightly adapts the $\mathsf{MSM}$ scheduler during the transition phase. Indeed, the latter would act as if $J_{i,k}$ did not complete its execution after $c_{i,k}$ time units, and take advantage of the unused resources to carry on the execution of the rem-jobs during exactly $\texttt{ur}_{i,k}(\ell)$ time units. Upon an $\texttt{IMCR}_{\ell+1}$, recovering these unused resources to carry on the execution of the rem-jobs of criticality equal to $\ell$ will preserve both the ($\ell+1$)-periodicity and ($\ell+1$)-feasibility, as proved by the following theorem.
\begin{theorem}
Upon an $\texttt{IMCR}_{\ell+1}$, the protocol that reclaims the unused resources $\texttt{ur}_{i,k}(\ell+1)$ of every job $J_{i,k}$ of criticality at least equal to $\ell+1$, to complete the execution of the rem-jobs of criticality $\ell$, is a valid increasing transition protocol.
\end{theorem}
\begin{IEEEproof}
The proof relies on the observation that reclaiming exactly $\texttt{ur}_{i,k}(\ell+1)$ units of computation from task $\tau_i$ consists in simulating a basic task scenario for $\tau_i$. Indeed, $c_{i,k} + \texttt{ur}_{i,k}(\ell+1) = C_i(\ell+1)$. Since the system was deemed $\mathsf{MSM}$-schedulable, it follows that any basic scenario is feasible. As a consequence, the proposed increasing transition protocol is valid.
\end{IEEEproof}

\subsubsection{A WCRT-Based Approach}\label{sec:improvedApproach}~\\
As explained in Section~\ref{sec:wcrt}, an upper-bound on the WCRT $R_i(\ell)$ of task $\tau_i\in\tau$ at criticality level $\ell$ is computed assuming an upper-bound on the interference suffered by $\tau_i$ in the most pessimistic scenario of criticality level $\ell$. Furthermore, in the most pessimistic scenario of criticality level $\ell$, the job $J_{i,k}$ should execute for exactly $C_i(\ell)$ which is also an upper-bound on the execution time of task $\tau_i$. Due to this twofold source of pessimism, it is consequently most unlikely that a job $J_{i,k}$ released by task $\tau_i$ will ever complete its execution exactly $R_i(\ell)$ time units after $r_{i,k}$. However, the system having been deemed $\mathsf{MSM}$-schedulable, it follows that the timeliness of each of its tasks could be respected in presence of this pessimism. In this section, we will thus suggest an increasing transition protocol which will enable for the reclaiming of unused resources, this time being based on the WCRT of each task. Our approach is based on the fact that the offline analysis assumed that a job $J_{i,k}$ could complete its execution $R_i(\ell)$ time units after $r_{i,k}$ in any scenario of criticality level $\ell$. If, however, $J_{i,k}$ completes its execution earlier, the scheduler will simulate the availability of $J_{i,k}$ until time $r_{i,k}+R_i(\ell)$, by assigning a processor to the rem-jobs whenever $J_{i,k}$ would have been executed. The challenge of this section is to prove that this will not increase the WCRT of tasks having a lower priority.
The following theorem formally proves that if every job released by tasks $\tau_j\in\texttt{hp}(\tau_i)$ completes its execution no later than $R_{j}(\ell)$ time units after its release, then every job released by task $\tau_i$ will complete its execution no later than $R_i(\ell)$ time units after its release as well, in any scenario of criticality level $\ell$.

\begin{theorem}\label{theo:mainTheo}
Each job $J_{i,k}$ released by a task $\tau_i$ will complete its execution no later than $R_i(\ell)$ time units after its release in any scenario of criticality level $\ell \leq L_i$, if for every task $\tau_j\in\texttt{hp}(\tau_i)$, every job released by a task $\tau_j$ completes its execution no later than $R_j(\ell)$ time units after its release.
\end{theorem}
\begin{IEEEproof}
The proof is by induction on the priorities of the tasks. We will indeed show that if the property holds for every task up to $\tau_{i-1}$, then it must also hold for task $\tau_i$.\\
\textbf{Base case:} The base case consists in considering task $\tau_1$, i.e.\@ the task having the highest priority. Since $\tau_1$ suffers no interference at all, it is trivial that every job released by $\tau_1$ will complete its execution no later than $R_1(\ell)=C_1(\ell)$ time units after its release, in any scenario of criticality level $\ell$.\\
\textbf{Induction step:} Assume that the property holds for every task up to $\tau_{i-1}$ and let us show that it must also hold for $\tau_i$, by considering the execution of any job $J_{i,k}$ released by $\tau_i$ over an interval of length $\Delta$. We will distinguish between two cases:\\
\textbf{Case 1:} let us first consider that there are no more than $m-2$ carry-in tasks belonging to $\texttt{hp}(i-1)$ over the interval of length $\Delta$. From Equation~\ref{eq:upperBoundTotIntWorkload}, according to the inductive hypothesis and since there are at most $m-2$ carry-in tasks belonging to $\texttt{hp}(\tau_{i-1})$, the actual total interfering workload suffered by task $\tau_{i-1}$ over the interval of length $\Delta$ can be expressed as follows:
			\begin{multline}\label{eq:case1:step1}
				\inter_{i-1}^*(\Delta,\ell) \leq \sum\limits_{\tau_j\in\texttt{hp}(\tau_{i-1})}\nctiw_{j,i-1}(\Delta,\ell) \\ 
					+ \sum_{\tau_j\in\texttt{hp}_{m-2}(\tau_{i-1})}\difftiw_{j,i-1}(\Delta,\ell)
			\end{multline}
			Because carry-in tasks have a higher interfering workload than non-carry-in tasks, an upper-bound on the interfering workload cause by task $\tau_{i-1}$ on task $\tau_i$ over the interval of length $\Delta$ is given by the following equation:
			\begin{equation}\label{eq:case1:step2}
				\inter_{i-1,i}^*(\Delta, \ell) \leq \nctiw_{i-1,i}(\Delta, \ell) + \difftiw_{i-1,i}(\Delta, \ell)
			\end{equation}
			Since the total interfering workload suffered by $\tau_i$ is equal to the total interfering workload suffered by $\tau_{i-1}$, plus the interfering workload suffered by $\tau_i$ from $\tau_{i-1}$, using Equations~\ref{eq:case1:step1} and~\ref{eq:case1:step2}, the actual total interfering workload suffered by job $J_{i,k}$ over the interval of length $\Delta$ can be upper-bounded as follows:
			\begin{equation}
				\inter_i^*(\Delta, \ell) \leq \inter_{i-1}^*(\Delta,\ell) + \nctiw_{i,i-1}(\Delta, \ell) + \difftiw_{i,i-1}(\Delta,\ell)
			\end{equation}
			However:
			\begin{multline}
				\difftiw_{i,i-1}(\Delta,\ell) + \sum_{\tau_j\in\texttt{hp}_{m-2}(\tau_{i-1})} \difftiw_{j,i-1}(\Delta,\ell) \\
					\leq \sum_{\tau_j\in\texttt{hp}_{m-1}(\tau_{i})} \difftiw_{j,i}(\Delta,\ell)
			\end{multline}
			Using Equation~\ref{eq:upperBoundTotIntWorkload}, we are able to conclude that:
			\begin{multline}\label{eq:case1:step5}
				\inter_i^*(\Delta,\ell) \leq \sum_{\tau_j\in\texttt{hp}(\tau_i)} \nctiw_{j,i}(\Delta,\ell) \\
					+ \sum_{\tau_j\in\texttt{hp}_{m-1}(\tau_{i})} \difftiw_{j,i}(\Delta,\ell) = \upinter_i(\Delta, \ell)
			\end{multline}
			Equation~\ref{eq:case1:step5} proves that the actual total interfering workload suffered by $J_{i,k}$ will not exceed $\upinter_i(\Delta, \ell)$, thus implying that $J_{i,k}$ will complete its execution no later than $R_i(\ell)$ time units after $r_{i,k}$.\\
\textbf{Case 2:} let us now consider the case where there are \emph{at least} $m-1$ carry-in tasks belonging to $\texttt{hp}(i-1)$ over the interval of length $\Delta$. From Property~\ref{property:guan}, since the total interfering workload is upper-bounded when considering at most $m-1$ carry-in tasks, $\tau_{i-1}$ can be considered as a non carry-in task. An upper-bound on the interfering workload suffered by $J_{i,k}$ from task $\tau_{i-1}$ is thus given by:
			\begin{equation*}
				\upinter_{i-1,i}(\Delta,\ell) \leq \nctiw_{i-1,i}(\Delta,\ell)
			\end{equation*}
		As a consequence, the actual total interfering workload suffered by job $J_{i,k}$ over the interval of length $\Delta$, can be expressed as:
			\begin{equation}\label{eq:step2}
				\inter_i^*(\Delta,\ell) \leq \inter_{i-1}^*(\Delta,\ell) + \nctiw_{i-1,i}(\Delta,\ell)
			\end{equation}
		According to the inductive hypothesis, we have:
			\begin{multline}\label{eq:step3}
				\inter_i^*(\Delta,\ell) \leq \sum\limits_{\tau_j\in\texttt{hp}(\tau_i)}\nctiw_{j,i}(\Delta,\ell) \\
					+ \sum_{\tau_j\in\texttt{hp}_{m-1}(\tau_{i-1})}\difftiw_{j,i}(\Delta,\ell) 
			\end{multline}
		However, since $\texttt{hp}(\tau_{i-1}) \subset \texttt{hp}(\tau_i)$, we have:
			\begin{equation}\label{eq:step4}
				\sum_{\tau_j\in\texttt{hp}_{m-1}(\tau_{i-1})}\difftiw_{j,i}(\Delta,\ell) \leq \sum_{\tau_j\in\texttt{hp}_{m-1}(\tau_{i})}\difftiw_{j,i}(\Delta,\ell)
			\end{equation}\label{eq:case2:step5}
		By injecting Equation~\ref{eq:step4} into Equation~\ref{eq:step3}, and by using Equation~\ref{eq:upperBoundTotIntWorkload} we get the following upper-bound on the actual total interfering workload suffered by $J_{i,k}$:
			\begin{multline}
				\inter_i^*(\Delta,\ell) \leq \sum\limits_{\tau_j\in\texttt{hp}(\tau_i)}\nctiw_{j,i}(\Delta,\ell) \\
					+ \sum_{\tau_j\in\texttt{hp}_{m-1}(\tau_{i})}\difftiw_{j,i}(\ell) = \upinter_i(\Delta,\ell)
			\end{multline}
		Equation~\ref{eq:case2:step5} proves that the actual total interfering workload suffered by $J_{i,k}$ will not exceed $\upinter_i(\Delta, \ell)$, implying that $J_{i,k}$ will complete its execution after no more than $R_i(\ell)$ time units.
\end{IEEEproof}

\begin{corollary}\label{corollary:main}
Each job $J_{i,k}$ released by a task $\tau_i$ will complete its execution no later than $R_i(\ell)$ time units after $r_{i,k}$ in any scenario of criticality level $\ell \leq L_i$, if every job $J_{j,p}$ released by a task $\tau_j \in \texttt{hp}(\tau_i)$ completes its execution exactly $R_j(\ell)$ time units after $r_{j,p}$.
\end{corollary}

From Corollary~\ref{corollary:main}, we propose an alternative increasing transition protocol, that slightly adapts the $\mathsf{MSM}$ scheduler during the transition phase resulting from the triggering of $\texttt{IMCR}_{\ell+1}$. Upon the completion of a job $J_{i,k}$ released by a task $\tau_i$ such that $L_i\geq\ell+1$ at time $f_{i,k} < r_{i,k}+R_i(\ell+1)$, the scheduler would simulate the availability of $J_{i,k}$ in the interval $[f_{i,k}, r_{i,k}+R_i(\ell+1)]$, thus acting as if $J_{i,k}$ had suffered an interference equal to the upper-bound $\mathcal{I}_i^\ell(R_i(\ell))$, and had completed its execution exactly at time $r_{i,k}+R_i(\ell)$. The following theorem proves that this approach will preserve both the ($\ell+1$)-periodicity and ($\ell+1$)-feasibility.

\begin{theorem}
Upon an $\texttt{IMCR}_{\ell+1}$, the protocol that simulates the availability of each job $J_{i,k}$ such that $J_i\geq\ell+1$ until time $r_{i,k}+R_i(\ell+1)$, to complete the execution of the rem-jobs of criticality $\ell$, is a valid increasing transition protocol.
\end{theorem}
\begin{IEEEproof}
This is an immediate consequence of Theorem~\ref{theo:mainTheo}. Indeed, by simulating the availability of each job $J_{i,k}$ released by task $\tau_i$ until time $r_{i,k}+R_i(\ell+1)$, we simulate the upper-bound on the interference suffered by $J_{i,k}$. According to Corollary~\ref{corollary:main}, this will not increase the WCRT of tasks $\tau_j\in\texttt{lp}(\tau_i)$. As a consequence, the proposed increasing transition protocol is valid.
\end{IEEEproof}

\subsubsection{Conclusion and Discussion}~\\
In this section, we have proposed several approaches to handle the transition from criticality level $\ell$ to criticality level $\ell+1$. The proposed approaches have a common goal, as they all allow to reach a compromise between a high safety, which is implied by the high level of pessimism adopted during the offline analysis phase, and an efficient usage of the platform, by recovering the unused resources to execute the less critical tasks during the online phase. In our opinion, this is an important aspect of our work, as it highlights the crucial fact that a high level of assurance does not necessarily imply that resources are doomed to be wasted.\\
Finally, note that the proposed approaches do not assume anything on the execution order of the rem-jobs. We will thus briefly discuss what management policy can be implemented in the case the system switches from operating mode $M_\ell$ to operating mode $M_{\ell+1}$ at time $t_{\texttt{IMCR}_{\ell+1}}$ while some rem-jobs of criticality equal to $\ell-1$ are still available. This means that at time $t_{\texttt{IMCR}_{\ell+1}}$, the system has to complete the execution of rem-jobs the criticality of which is less than $\ell$. In that case, either it is assumed that since their criticality is less than the current criticality of the system, they are all equal in terms of importance, meaning that a global strategy can be applied to complete their execution (complete the rem-jobs with the shortest deadline first, complete the rem-jobs with the shortest remaining processing time first, etc.). Or we can assume that even though the current criticality of the system is higher than their own criticality, the relative importance of the rem-jobs remains the same, meaning it is safer to complete the rem-jobs in decreasing order of criticality level (the rem-jobs with a criticality equal to $\ell$ have to be completed before the rem-jobs with a criticality equal to $\ell-1$).

\subsection{Handling Decreasing Mode Change Requests}\label{subsec:decreasing}
In our previous work~\cite{Santy:2012:RMS:2354411.2355225}, we focused on \emph{uni}processor platforms, and proved that whenever an idle time was detected, while the system had reached criticality level $\ell > 1$, it was safe to re-enable \emph{every} task that had previously been suspended. However, the occurrence of a simultaneous idle time on every processors of a multiprocessor platform is unlikely. In this section, we will therefore suggest an alternative approach to re-enable the suspended tasks without relying on the occurrence of idle times. Before going any further, let us introduce the following property, which highlights two important requirements regarding task re-enablement. \\
\begin{property}\label{prop:dmcr:safe}
When the system is switching from operating mode $M_{h}$ to operating mode $M_{\ell}$ upon a $\texttt{DMCR}_{\ell}$ at time $t_{\texttt{DMCR}_{\ell}}$, the re-enablement of a previously suspended task $\tau_i$ belonging to $M_{\ell}$ at time $t \geq t_{\texttt{DMCR}_{\ell}}$ may be carried out provided that the following two conditions hold:
\begin{compactitem}
	\item[$\bullet$] The re-enablement at time $t$ of task $\tau_i$ must not jeopardize the $h$-periodicity and $h$-feasibility of the system;
	\item[$\bullet$] Task $\tau_i$ must be guaranteed to meet its deadline from time $t$ onward if no $\texttt{IMCR}_{\ell+1}$ occurs at a time $t_{\texttt{IMCR}_{\ell+1}} \geq t$.
\end{compactitem}
\end{property}
In the rest of this paper, we will say that it is \emph{safe} to re-enable a task $\tau_i$ if Property~\ref{prop:dmcr:safe} holds upon the re-enablement of $\tau_i$. The way new-mode tasks are re-enabled leads to distinguish two types of protocols.
\begin{definition}[Synchronous/Asynchronous protocol~\cite{Real:2004:MCP:969960.969963}]\label{def:sync:async}
Assuming the system is switching from operating mode $M_{h}$ to operating mode $M_{\ell}$, a \emph{synchronous} protocol is a protocol that re-enables each suspended task belonging to operating mode $M_{\ell}$ simultaneously. An \emph{asynchronous} protocol is a protocol that enables every suspended task belonging to operating mode $M_{\ell}$ independently from the others, i.e.\@ some tasks belonging to $M_{\ell}$ might be enabled earlier than others.
\end{definition}
\begin{figure}
	\centering
	\includegraphics[scale=0.39]{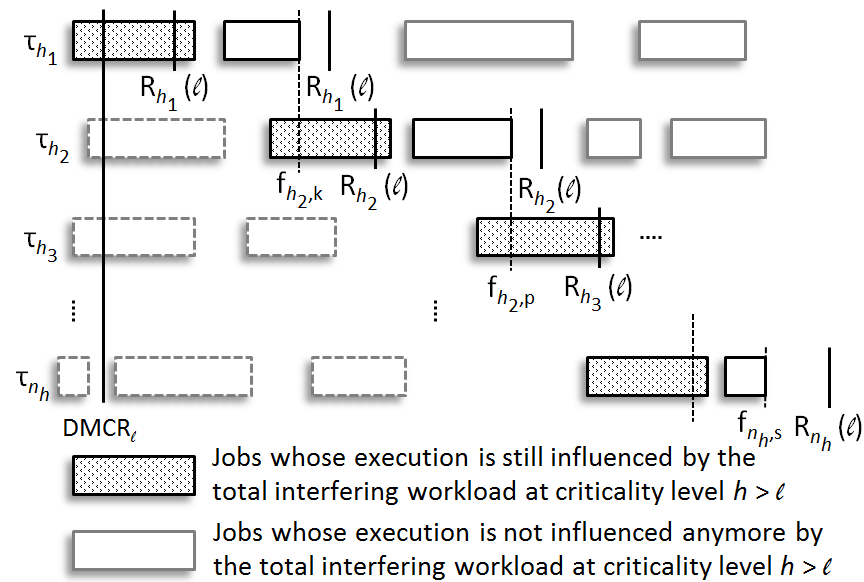}
	\caption{The decreasing mode transition protocol iteratively identifies the time instants $f_{i,k} \leq R_i(\ell)$ for every task $\tau_i\in\tau^h$.}\label{figure:dmcr}
\end{figure}
Assuming the system is executing in operating mode $M_h$, and depending on the set of functionalities the system is willing to re-enable, once the rem-jobs of criticality less than $h$ are completed, a $\texttt{DMCR}_\ell$, $1\leq\ell<h$ can be triggered. For the sake of clarity, we will assume that no $\texttt{IMCR}_{\ell+1}$ is triggered during the decreasing mode change (we will shortly discuss this potential event later). We suggest a synchronous protocol, illustrated by Figure \ref{figure:dmcr}, that works as follows: starting at time $t_{\texttt{DMCR}_\ell}$, the protocol identifies the first job $J_{h_1,k}$ released by task $\tau_{h_1}$ that completes its execution at time $f_{h_1,k} \leq r_{h_1,k} + R_{h_1}(\ell)$. From time $f_{h_1,k}$, the protocol then identifies the first job $J_{h_2,p}$ released by task $\tau_{h_2}$ that completes its execution at time $f_{h_2,p}\leq r_{h_2,p}+R_{h_2}(\ell)$. The procedure then keeps on identifying one such job for each task $\tau_i\in\tau^h$, in order of their priority (i.e.\@ it identifies such a job for task $\tau_{h_i}$ only when it has previously identified such a job for task $\tau_{h_{i-1}}$). The re-enablement of the previously suspended tasks having a criticality at least equal to $\ell$ can then take place when the procedure identifies a job $J_{n_h, s}$ that completes its execution at time $f_{n_h, s} \leq r_{n_h,s}+R_{n_h}(\ell)$. In the following, we will assume that if the protocol identifies a job $J_{h_i,k}$ that completes its execution at time $f_{h_i,k} \leq r_{h_i,k}+R_i(\ell)$, then this means that the procedure has already identified such a job for every tasks $\tau_{h_1}, \tau_{h_2}, ...,\tau_{h_{i-1}}$.

\begin{lemma}\label{lemma:decreasing:interferingworkload}
Whenever a task $\tau_{h_i}\in\tau^h$ releases a job $J_{h_i,k}$ that completes its execution at time $f_{h_i,k}\leq r_{h_i,k}+R_{h_i}(\ell)$, and no $\texttt{IMCR}_{\ell+1}$ is triggered during the mode change, then the actual interfering workload suffered by tasks $\tau_j\in\texttt{lp}(\tau_{h_i})$ from task $\tau_{h_i}$, from time $f_{h_i,k}$ onward will be less than or equal to $\citiw_{h_i,j}(\Delta, \ell)$ (resp. $\nctiw_{h_i,j}(\Delta, \ell)$) over any window of length $\Delta$, if $\tau_{h_i}$ is a carry-in (resp. non carry-in) task for $\tau_j$.
\end{lemma}
\begin{proof}
The upper-bound on the interfering workload suffered by task $\tau_j\in\texttt{lp}(\tau_{h_i})$ from task $\tau_{h_i}$ over an interval of length $\Delta$, in any $\ell$-interval, is computed assuming $\tau_{h_i}$ will release jobs that will execute for no more than $C_{h_i}(\ell)$ time units. However, if no $\texttt{IMCR}_{\ell+1}$ is triggered, then it must be the case that every job released by task $\tau_{h_i}$ executes for no more than $C_{h_i}(\ell)$ time units. Therefore, from $f_{h_i,k}$ onward, the actual total interfering workload suffered by tasks $\tau_j\in\texttt{lp}(\tau_{h_i})$ from task $\tau_{h_i}$ will be less than or equal to $\citiw_{h_i,j}(\Delta, \ell)$ (resp. $\nctiw_{h_i,j}(\Delta, \ell)$) over any window of length $\Delta$ if $\tau_{h_i}$ is a carry-in (resp. non carry-in) task for $\tau_j$.
\end{proof}

\begin{lemma}\label{lemma:decreasing:totalInterferingWorkload}
Let us assume the protocol identified a job $J_{h_{i-1},k}$ that completed its execution at time $f_{h_{i-1},k} \leq r_{h_{i-1},k}+R_{h_{i-1}}(\ell)$. From time $f_{h_{i-1},k}$ onward the actual total interfering workload $\inter_i^*(\Delta,\ell)$ suffered by task $\tau_{h_i}$ over any window of size $\Delta$ will be less than or equal to $\upinter_i(\Delta,\ell)$.
\end{lemma}
\begin{proof}
From Lemma~\ref{lemma:decreasing:interferingworkload}, we know that from time $f_{h_{i-1},k}$ onward, the actual interfering workload of every task $\tau_j\in\texttt{hp}(\tau_{h_i})$ on $\tau_i$ is less than or equal to $\citiw_{j,h_i}(\Delta,\ell)$ or $\nctiw_{j,h_i}(\Delta,\ell)$, depending whether $\tau_j$ is a carry-in task for $\tau_{h_i}$ or not. Furthermore, from Equation \ref{eq:upperBoundTotIntWorkload}, we know that an upper-bound on the total interfering workload suffered by task $\tau_{h_i}$ is given by summing $\nctiw_{j,h_i}(\Delta,\ell)$ $\forall\tau_j\in\texttt{hp}(\tau_{h_i})$ with the $m-1$ largest values of $\difftiw_{j,h_i}(\Delta,\ell)$. Thus, the actual total interfering workload suffered by $\tau_{h_i}$ will be less than or equal to $\upinter_{h_i}(\Delta,\ell)$.
\end{proof}

\begin{corollary}\label{decreasing:corollaryTotIntWork}
Upon a $\texttt{DMCR}_\ell$, when the protocol identifies a job $J_{n_h,s}$ that completed its execution at time $f_{n_h,s} \leq r_{n_h,s}+R_{n_h}(\ell)$, then from time $f_{n_h,s}$ onward, the actual total interfering workload suffered by any task $\tau_i\in\tau$ is less than or equal to $\upinter_i(\Delta,\ell)$.
\end{corollary}

From Corollary \ref{decreasing:corollaryTotIntWork}, we will now show that it is safe to re-enable every suspended task $\tau_k\in\tau^\ell$ at time $f_{n_h,s}$.

\begin{lemma}\label{lemma:decreasing:safeNewModeTasks}
Upon a $\texttt{DMCR}_\ell$, when protocol identifies a job $J_{n_h,s}$ that completed its execution at time $f_{n_h,s} \leq r_{n_h,s}+R_{n_h}(\ell)$, then the jobs released by every suspended task $\tau_k\in\tau^\ell$ from time $f_{n_h,s}$ onward will meet their deadlines in any scenario of criticality level at most $L_k$.
\end{lemma}
\begin{proof}
The WCRT at criticality level $\ell$ of each task $\tau_k\in\tau^\ell$ was computed assuming a total interfering workload of $\upinter_i(\Delta,\ell)$. But from Corollary \ref{decreasing:corollaryTotIntWork}, we know that at time $f_{n_h,s}$, the actual total interfering workload suffered by any task $\tau_i$ in the system is less than or equal to $\upinter_i(\Delta,\ell)$. Since the system was deemed $\mathsf{MSM}$-schedulable, every job released by a task $\tau_k\in\tau^\ell$ from time $f_{n_h,s}$ onward will meet its deadline in any scenario of criticality level at most $L_k$.
\end{proof}

\begin{lemma}\label{lemma:decreasing:safeOldModeTasks}
Assume the system is executing in operating mode $M_h$. Upon a $\texttt{DMCR}_\ell$, when protocol identifies a job $J_{n_h,s}$ that completed its execution at time $f_{n_h,s} \leq r_{n_h,s}+R_{n_h}(\ell)$, the re-enablement of every suspended task $\tau_k\in\tau^\ell$ will not jeopardize the $h$-periodicity and $h$-feasibility.
\end{lemma}
\begin{proof}
From Corollary \ref{decreasing:corollaryTotIntWork}, we know that at time $f_{n_h,s}$, the actual total interfering workload suffered by any task $\tau_i$ in the system is less than or equal to $\upinter_i(\Delta,\ell)$. However, this upper-bound is computed assuming that \emph{every} task $\tau_k\in\tau^\ell$ will release jobs that complete their execution after no more than $C_i(\ell)$ time units. Therefore, the actual total interfering workload suffered by any task $\tau_j$ upon the re-enablement of every suspended task $\tau_k\in\tau^\ell$ will still be less than or equal to $\upinter_i(\Delta,\ell)$. Furthermore, since the system was deemed $\mathsf{MSM}$-schedulable, it follows that every job released by tasks $\tau_j\in\tau^h$ will meet its deadline in any scenario of criticality up to level $L_j\geq h$. It follows that both the $h$-periodicity and $h$-feasibility are preserved upon the re-enablement of every suspended task $\tau_k\in\tau^\ell$.
\end{proof}

\begin{theorem}\label{theo:decreasing:mainResult}
Let us assume the system is executing in operating mode $M_h$. Upon a $\texttt{DMCR}_\ell$, when the protocol identifies a job $J_{n_h,s}$ that completed its execution at time $f_{n_h,s} \leq r_{n_h,s}+R_{n_h}(\ell)$, it is safe to re-enable the suspended task belonging to the operating mode $M_\ell$ at time $f_{n_h,k}$.
\end{theorem}
\begin{IEEEproof}
To prove this theorem, we have to show that Property \ref{prop:dmcr:safe} holds upon the re-enablement of every suspended task $\tau_j\in\tau^{\ell}$ at time $f_{n_h,k}$. Lemma \ref{lemma:decreasing:safeNewModeTasks} proved that every suspended task $\tau_i$ belonging to the operating mode $M_\ell$ could release jobs that would be able to meet their deadline if $\tau_i$ was re-enabled at time $f_{n_h,k}$, and provided that no $\texttt{IMCR}_{\ell+1}$ was triggered. Furthermore, Lemma \ref{lemma:decreasing:safeOldModeTasks} proved that the re-enablement of every suspended task $\tau_i$ belonging to the operating mode $M_\ell$ would not jeopardize the $h$-periodicity and $h$-feasibility. It follows that at time $f_{n_h,k}$, it is safe to re-enable every suspended task belonging to the operating mode $M_\ell$.
\end{IEEEproof}
Theorem \ref{theo:decreasing:mainResult} proves that it will eventually be possible to decrease the criticality level of the system, provided the computational demand of the tasks having a criticality higher than or equal to $h$ decreases. It follows that the suspension delay suffered by tasks having a criticality less than $h$ is reduced. In practice however, if an $\texttt{IMCR}_{\ell+1}$ is triggered during during the $\texttt{DMCR}_\ell$ handling, then the procedure that consists in finding the first job $J_{{n_h},s}$ that completes its execution a time $f_{{n_h},s}\leq r_{n_h,s}+R_{n_h}(\ell)$ is aborted. Indeed, in that case, we can no longer guarantee that the additional interfering workload generated by tasks having a criticality equal to $\ell$ will not jeopardize the $(\ell+1)$-periodicity and $(\ell+1)$-feasibility.
\begin{corollary}
Upon a $\texttt{DMCR}_\ell$, reenabling every suspended task belonging to operating mode $M_\ell$ upon the identification of a job $J_{n_h,s}$ that completed its execution at time $f_{n_h,s} \leq r_{n_h,s}+R_{n_h}(\ell)$ is a valid decreasing transition protocol.
\end{corollary}

\section{Conclusion}\label{sec:conclusion}
In this work, the first contribution consisted in formalizing mixed-criticality systems in terms of multi-moded systems, thus giving a new outlook to the problem. As a second contribution, we have shown that multi-moded approaches could help solve the consistency problems that arise when less critical tasks are brutally discarded, by enabling for a softer switch from a lower criticality level to a higher one. Finally, as a third contribution, we have highlighted the fact that the behavior of such systems could be greatly enhanced, by proving that task re-enablement was possible without compromising its safety. Those approaches allow to achieve a much more adept usage of the platform, by avoiding to vainly waste computational resources. Future work will concentrate, among other things, on suggesting asynchronous decreasing transition protocols, which allow for a progressive re-enablement of the less critical tasks.

\nocite{*}
\bibliographystyle{IEEEtran}
\bibliography{bibliographie}

\end{document}